\pdfoutput=1

\documentclass[11pt,letterpaper]{article}
\renewcommand{\footnoterule}{%
  \kern -3pt
  \hrule width \textwidth height 0.5pt
  \kern 2pt
}
\usepackage{latexsym}
\usepackage{amssymb}
\usepackage{amsmath}
\usepackage{amsthm}
\usepackage{graphicx}
\usepackage{bm}
\usepackage{enumerate}
\usepackage{sectsty}
\usepackage{subfig}
\usepackage{float}
\usepackage{caption}
\usepackage{wasysym}
\usepackage{times,fancyhdr}
\usepackage{mathrsfs}
\usepackage{setspace}
\sectionfont{\normalsize}
\subsectionfont{\small}

\newcommand{\romansubs}{\renewcommand{\theequation}{\theparentequation \roman{equation}}}

\newcommand{\PRLsep}{\noindent\makebox[\linewidth]{\resizebox{0.750\linewidth}{1pt}{$\blacklozenge$}}\bigskip}
\newcommand{\PRLsepsmall}{\noindent\makebox[\linewidth]{\resizebox{0.40\linewidth}{0.5pt}{$\blacklozenge$}}\bigskip}
\newcommand*{\Scale}[2][4]{\scalebox{#1}{$#2$}}%
\newcommand{\ect}{Einstein-Cartan theory }
\newcommand{\ecg}{Einstein-Cartan gravity }
\newcommand{\ectend}{Einstein-Cartan theory}
\newcommand{\ecgend}{Einstein-Cartan gravity}
\newcommand{\weys}{Weyssenhoff }

\newcommand{\tinyrmsub}[1]{\mbox{{\tiny{#1}}}}
\newcommand{\christoffel}[3]
{\ensuremath{\left\{\hspace{-0.15cm}\! 
\begin{array}{l}
{#1} \\ 
{#2}\,{#3}
\end{array}
\!\hspace{-0.15cm}\right\}}}

\newtheorem*{lemma}{Lemma}

\DeclareMathOperator{\sech}{sech}

\begin{document}

\pagestyle{fancy}
\fancyhead{} % clear all header fields
\fancyhead[OR]{\thepage}
\fancyhead[OC]{{\footnotesize{\textsf{ENERGY CONDITION RESPECTING WARP DRIVES IN EINSTEIN-CARTAN THEORY}}}}
%%\fancyhead[OC]{{\footnotesize{ENERGY CONDITION RESPECTING WARP DRIVES IN EINSTEIN-CARTAN THEORY}}}
\fancyfoot{} % clear all footer fields
\renewcommand\headrulewidth{0.5pt}
\addtolength{\headheight}{2pt} % make space for the rule

\title{\bf{\large{\textsf{Energy Condition Respecting Warp Drives: The Role of Spin in Einstein-Cartan Theory}}}}
%%\title{{\small \bf{ENERGY CONDITION RESPECTING WARP DRIVES: THE ROLE OF SPIN IN EINSTEIN-CARTAN THEORY}}}
\author {{\small Andrew DeBenedictis \footnote{adebened@sfu.ca}} \\
\it{\small The Pacific Institute for the Mathematical Sciences} \\
\it{\small and}\\
\it{\small Department of Physics, Simon Fraser University,}\\
\it{\small Burnaby, British Columbia, V5A 1S6, Canada }
\\[-0.1cm]
%%\line(1,0){45}\\[0.1cm]
\PRLsepsmall\\[-0.5cm]
\and
{\small Sa{\v s}a Iliji{\'c} \footnote{sasa.ilijic@fer.hr}} \\
\it{\small Department of Applied Physics, Faculty of Electrical Engineering and Computing, University of Zagreb}\\
\it{\small HR-10000 Zagreb, Unska 3, Croatia }
}
\date{{\small September 7, 2018}}
\maketitle

\setcounter{footnote}{0}
\begin{abstract}
\noindent In this paper we study the so called ``warp drive'' spacetimes within the $U_{4}$ Riemann-Cartan manifolds of \ectend. Specifically, the role that spin may play with respect to energy condition violation is considered. It turns out that with the addition of spin, the torsion terms in \ecg do allow for energy condition respecting warp drives. Limits are derived which minimize the amount of spin required in order to have a weak/null-energy condition respecting system. This is done both for the traditional Alcubierre warp drive as well as for the modified warp drive of Van Den Broeck which minimizes the amount of matter required for the drive. The ship itself is in a region of effectively vacuum and hence the torsion, which in \ect is localized in matter, does not affect the geodesic nature of the ship's trajectory. We also comment on the amount of spin and matter required in order for these conditions to hold.
\end{abstract}
\rule{\linewidth}{0.2mm}
\vspace{-1mm}

\section{{Introduction}}
There is no doubt that general relativity, with its description of a dynamical spacetime, is one of the most fascinating of physical theories. In its over one-hundred year history it has changed our understanding of the universe dramatically. For example, general relativity has provided an explanation for the residual perihelion precession of the planets \cite{ref:perihelion}, and has predicted an expanding universe \cite{ref:expandingstart}, \cite{ref:expandingend}. To date, general relativity, with the introduction of dark matter and dark energy, has passed experimental tests of very high precision \cite{ref:willbook}. With the more recent direct detection of gravitational waves from black hole and neutron star events, the tests of general relativity are no longer restricted merely to the weak-field regime. It can easily be argued, therefore, that general relativity remains a robust theory of gravity. 

Because of these successes, any deviations of gravitational theory from general relativity are highly restricted. If the theory of classical gravity is not general relativity, then it must be very close to it, matching it almost exactly in the regimes in which gravity has been tested accurately. One rather interesting theory is that of \ecg \cite{ref:einstcart}, \cite{ref:einstcart2}. It is arguably the simplest extension of gravity which includes torsion and does not alter the dynamical spacetime picture of general relativity. In fact, the \ect may be viewed simply as general relativity supplemented with torsion. In \ecg the spacetime possesses torsion as well as curvature, and in the limit that torsion vanishes it coincides exactly with general relativity. One reason \ect has not been ruled out is that the effects of torsion in this theory are rather difficult to measure. It is the spin of matter which couples to the antisymmetric part of the connection, and the spin is directly proportional to the modified torsion tensor. Therefore, outside of matter, there are no torsion effects, and gravitation is fully governed by general relativity, although the vacuum solution may now differ somewhat compared to that of pure general relativity, due to the source term having been modified by the spin. The small magnitude of the spin-torsion coupling, along with the fact that experiments validating or ruling out torsion effects must be done within matter possessing significant spin content, means that performing experiments which may invalidate \ecg is very difficult. 

Under extreme conditions, however, it may be that the spin density of matter becomes large enough to produce serious deviations from pure general relativity. For example, in cosmology \ect has been shown to eliminate the big bang singularity \cite{ref:ectbigbang}. As well, \ecg may naturally explain the flatness
and horizon problems, due to the presence of small torsion densities \cite{ref:popinfl}. In the realm of black holes it has been shown that the torsion leads to a non-singular bounce in gravitational collapse that otherwise would lead to a singularity within general relativity, \cite{ref:ziaie},\cite{ref:hashemicollapse}. Regarding more exotic solutions, wormholes have been studied within \ect \cite{ref:bronwh}, \cite{ref:ectwhsols}. Other studies in spacetimes with torsion include Maxwell fields \cite{ref:katkarmax}, Proca fields \cite{ref:seitzproca}, and Dirac fields \cite{ref:ecdiracthesis}. (See also references therein.)

Many of the above studies indicate that torsion in \ect acts as a moderating effect in gravitation. That is, torsion often softens the effects of gravity, eliminating seemingly unphysical effects in gravitational theory. It is with this in mind that we study here what are known in the literature as warp drive spacetimes within \ectend. It is known that within general relativity (GR) warp drive spacetimes must violate energy conditions \cite{ref:alcub}, \cite{ref:lobobook} and therefore it is of interest to study if the energy condition violation may be eliminated by the torsion effects of \ecg. For this we utilize the \weys spin fluid description of matter \cite{ref:weys}, supplemented with spinless auxiliary structure, whose rationale is described later. The \weys description of matter has been rather successful in various studies in \ect \cite{ref:weysstudystart} - \cite{ref:weysstudyend} and it represents a fluid model whose spin content is manifest.

We choose to study the warp drive spacetimes for several reasons. Perhaps the primary reason is that it is generally interesting to study exactly what established theories may predict under extreme situations. There may be much that can be learned from studying the extreme limits of a physical theory. In this vein there is ample pedagogical value in studying exotic solutions to gravitational field theory. This, for example, was a primary motivation in the now classic paper of Morris and Thorne \cite{ref:morthorne} (see also \cite{ref:hiscock}). Also, exotic solutions are interesting to study in their own right, as they serve to illustrate the richness of the solution space of field theories. There is no doubt that the warp drive, though not practically feasible, is an interesting solution to the gravitational field equations in much the same way as, for example, the G\"{o}del universe is \cite{ref:godel}-\cite{ref:godelreview}.

%%% This paper is organized as follows

\section{A brief review of the \ect of gravity}\label{sec:ect}
We give here a short review of \ecg and of the \weys fluid. Units will be used such that $G=c=\hbar=1$, and these factors will be reinstated in the final analysis. 

In the microscopic realm of relativistic matter, the rotational sector of the Lorentz group naturally classifies elementary particles within unitary representations of the group's rotations. The spin notion of a particle is therefore just as elementary as its mass. It seems interesting then that matter's spin content, unlike its mass counterpart (and by extension energy and momentum) does not play a role as a source of gravity in general relativity theory. The motivation behind \ect is to eliminate this asymmetry between mass and spin and so also include spin as a source of gravitation. Interestingly, since spin is a fundamental quantum mechanical property of matter, it may be that a quantum theory of gravity must include a spin coupling to gravitation in order to be fully consistent. 

The key ingredient of \ect that causes deviations from Einstein gravity is the presence of non-zero torsion, $T_{\beta\gamma}^{\;\;\;\alpha}$, in the spacetime affine connection $\Gamma^{\alpha}_{\;\beta\gamma}$:
\begin{equation}
 T_{\beta\gamma}^{\;\;\;\alpha}:=\frac{1}{2}\left[\Gamma^{\alpha}_{\;\beta\gamma} - \Gamma^{\alpha}_{\;\gamma\beta} \right]\,. \label{eq:torsion}
\end{equation}
Since in general relativity the symmetric Christoffel connection, $\Gamma^{\alpha}_{\;\beta\gamma}\rightarrow${\tiny{$\christoffel{\alpha}{\beta}{\gamma}$}}, is utilized, torsion identically vanishes everywhere in Einstein gravity. It should be noted that although the connection is not a tensor, the torsion, being a difference of connections, is. Further, in order to make ties with the well tested theory of special relativity, it is demanded that the metric tensor is covariantly constant
\begin{equation}
 \nabla_{\mu}\,g_{\alpha\beta} = 0\, , \label{eq:covconstmet}
\end{equation}
where the covariant derivative in (\ref{eq:covconstmet}) is with respect to the full connection, $\Gamma^{\alpha}_{\;\beta\gamma}$. A manifold in which (\ref{eq:covconstmet}) holds is called a $U_{4}$ manifold. Further, if one restricts the torsion to zero one will have a Riemannian manifold, and if one instead restricts curvature, but not torsion, to zero, one has a Weitzenb\"{o}ck manifold. Here we are concerned with the Riemann-Cartan manifold, $U_{4}$, allowing for both curvature as well as torsion. Whereas upon transport around an infinitesimal closed loop, curvature yields a holonomy in the angle of the vector (see figure \ref{fig:holonomy}a), torsion on the other hand preserves the orientation of the vector but instead induces a holonomy in its translation in the tangent space (see figure \ref{fig:holonomy}b).  

\begin{figure}[h!t]
\begin{center}
\includegraphics[viewport=0 320 1495 865,  width=\textwidth, clip, keepaspectratio=true]{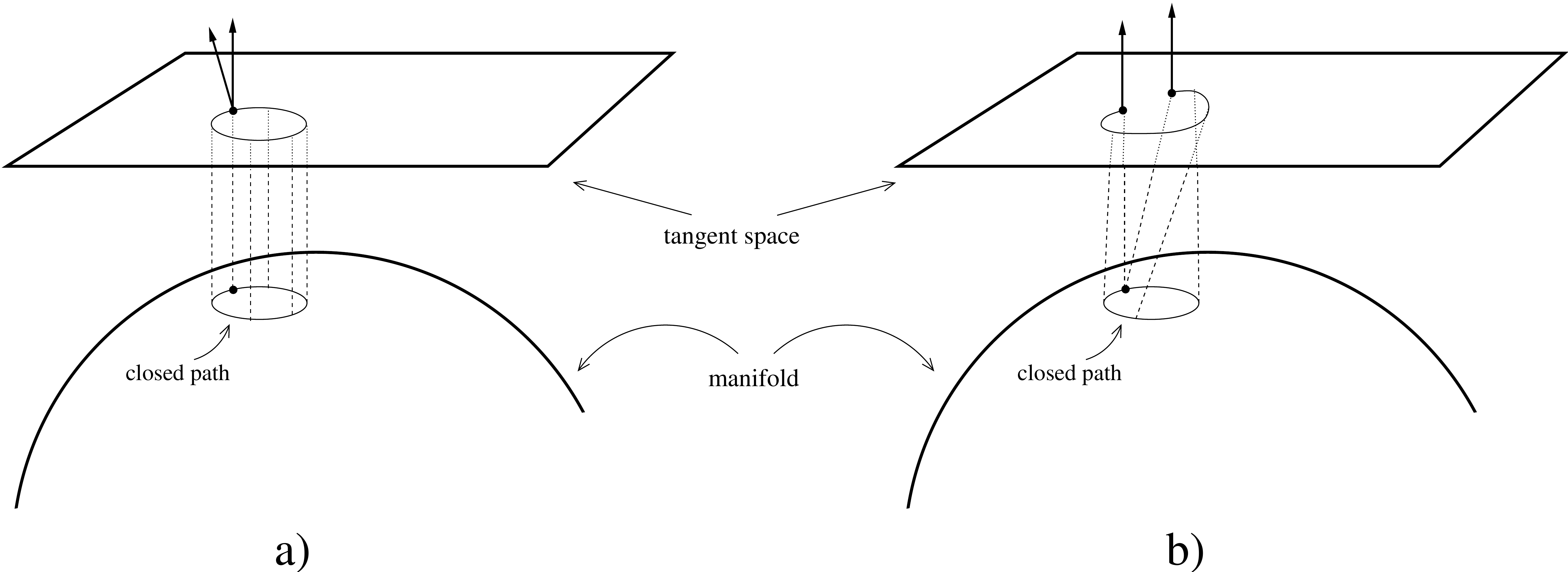}
\caption{{\small{A pictorial description of curvature and torsion: a) When the vector is parallel transported in a closed loop in a manifold with curvature, the vector possesses an angular defect in the tangent space upon return to the original point.\quad b) When the vector is parallel transported in a closed loop in a manifold with torsion, the vector acquires a translational defect in the tangent space upon return to the original point in the manifold. If the manifold possesses both curvature and torsion both an angular defect and a translational defect will be present.}}}
\label{fig:holonomy}
\end{center}
\end{figure}

In \ecg the action is postulated to resemble that of general relativity,
\begin{equation}
 I=\int d^{4}x\,\sqrt{-g} \left[-\frac{1}{2\kappa} R + \mathscr{L}_{\tinyrmsub{m}}\right]\,, \label{eq:action}
\end{equation}
with $\kappa=8\pi G/c^{4}$ (in our units $\kappa = 8 \pi$), $g$ the determinant of the metric tensor, and $\mathscr{L}_{\tinyrmsub{m}}$ the matter Lagrangian density. It should be noted that the curvature scalar $R$ is to be calculated with the full connection of the theory. Specifically, the general connection is given by
\begin{equation}
 \Gamma^{\alpha}_{\;\beta\gamma}={\tiny{\christoffel{\alpha}{\beta}{\gamma}}} - K_{\beta\gamma}^{\;\;\;\alpha}\,, \label{eq:connection}
\end{equation}
where the last quantity is known as the contorsion (sometimes contortion) tensor
\begin{equation}
 K_{\beta\gamma}^{\;\;\;\,\alpha}:=T_{\gamma\;\;\beta}^{\;\,\alpha}-T_{\beta\gamma}^{\;\;\;\,\alpha} - T^{\alpha}_{\;\,\beta\gamma}\,.  \label{eq:contortion}
\end{equation}
Since the connection is not symmetric, the Ricci tensor in \ect will also in general not be symmetric. One must of course assume here that the metric tensor and the connection are independent quantities.

In \ect the equations of motion are derived via variation of (\ref{eq:action}) with respect to the metric and also with respect to the contorsion tensor (\ref{eq:contortion}), which represents the non-metric part of the independent connection. This results in two sets of equations:

\begin{subequations} \romansubs
{\allowdisplaybreaks\begin{align}
G^{\mu\nu} - \left(\nabla_{\alpha} + 2T_{\alpha\lambda}^{\;\;\;\lambda}\right)\left[S^{\mu\nu\alpha} - S^{\nu\alpha\mu} +S^{\alpha\mu\nu}\right] = & \kappa \mathcal{T}^{\mu\nu}\,, \label{eq:eom1} \\[0.1cm]
S^{\alpha}_{\;\beta\gamma}= & \kappa \tau^{\alpha}_{\;\beta\gamma}\,, \label{eq:eom2}
\end{align}}
\end{subequations}
with $G_{\mu\nu}$ the Einstein tensor created out of the full (non-symmetric) connection, and $\mathcal{T}^{\mu\nu}$ the usual (symmetric) stress-energy tensor. The covariant derivative in (\ref{eq:eom1}) is also with respect to the full connection. The above two sets of equations constitute the Einstein-Cartan field equations. It should be mentioned here that the left-hand side of (\ref{eq:eom1}) is actually symmetric overall, as it must be in order to equal its right-hand side. In (\ref{eq:eom2}) the quantity $S^{\alpha}_{\;\beta\gamma}$ represents the modified torsion tensor, sometimes known as the superpotential:
\begin{equation}
 S_{\alpha\beta}^{\;\;\;\;\gamma}:=T_{\alpha\beta}^{\;\;\;\gamma}+\delta^{\gamma}_{\;\alpha}T_{\beta\;\;\,\lambda}^{\;\,\lambda}-\delta^{\gamma}_{\;\beta}T_{\alpha\;\;\,\lambda}^{\;\,\lambda}\,.
\end{equation}
In (\ref{eq:eom2}) there is also present the dynamical spin tensor, $\tau^{\alpha}_{\;\beta\gamma}$. This quantity is the spin analog of the stress-energy tensor. That is, it contains the physical spin content of the theory. 

By substituting equation (\ref{eq:eom2}) in lieu of the superpotential terms in (\ref{eq:eom1}) one may re-write (\ref{eq:eom1}) as 
\begin{equation}
 G^{\mu\nu}=R^{\mu\nu}-\frac{1}{2}R\,g^{\mu\nu}=  \kappa \Theta^{\mu\nu}\,, \label{eq:eomb}
\end{equation}
where the Ricci tensor is the non-symmetric tensor constructed out of the full connection, and $R$ is its trace. $\Theta_{\mu\nu}$ is the non-symmetric \emph{canonical} stress-energy tensor with spin content. The antisymmetric part of (\ref{eq:eomb}) is automatically satisfied and therefore (\ref{eq:eomb}) is equivalent to (\ref{eq:eom1}) and (\ref{eq:eom2}). Explicitly written, with all torsion and modified torsion terms replaced by spin tensors via (\ref{eq:eom2}), the surviving terms in (\ref{eq:eom1}) or (\ref{eq:eomb}) yield the following equation:
\begin{equation}
\Scale[0.90]{G^{\mu\nu}\left({\tiny{\christoffel{\alpha}{\beta}{\gamma}}}\right)-\kappa^{2}\left[\tau^{\lambda\sigma\mu}\tau_{\lambda\sigma}^{\;\;\;\nu} -2 \tau^{\mu\lambda\sigma}\tau^{\nu}_{\;\lambda\sigma} -4 \tau^{\mu\lambda}_{\;\;\;[\sigma}\tau^{\nu\sigma}_{\;\;\;\lambda]} +\frac{1}{2}g^{\mu\nu}\left( 4 \tau_{\rho\;\;\;[\sigma}^{\;\lambda} \tau^{\rho\sigma}_{\;\;\;\lambda]} + \tau^{\rho\lambda\sigma}\tau_{\rho\lambda\sigma}\right) \right] = \kappa \mathcal{T}^{\mu\nu}\,,} \label{eq:goodeom}
\end{equation}
with $G^{\mu\nu}\left({\tiny{\christoffel{\alpha}{\beta}{\gamma}}}\right)$ the Einstein tensor created from the Christoffel connection\footnote{One might expect to find derivatives of $\tau^{\alpha}_{\;\beta\gamma}$ in the resulting equation due to the covariant derivative in (\ref{eq:eom1}). However, these cancel with corresponding derivatives in the full $G^{\mu\nu}$ when the modified torsion terms in $G^{\mu\nu}$ are replaced with the spin tensor via (\ref{eq:eom2}).}. 

In \ect the algebraic structure of equation (\ref{eq:eom2}) dictates that the modified torsion tensor vanishes wherever the spin tensor vanishes. Therefore, any torsion effects are manifest only inside of matter. Outside of matter the theory is equivalent to general relativity. As well, outside of matter, test particles (meaning here particles whose stress-energy and spin may be neglected for gravitational purposes) will follow the usual geodesic equation of general relativity. These properties of \ecg will be particularly desirable for the study here.

The physics of the spin is contained in the tensor $\tau^{\alpha}_{\;\beta\gamma}$, in much the same way as the physics of energy is contained in the stress-energy tensor. One may, for example, prescribe certain components of $\tau^{\alpha}_{\;\beta\gamma}$ from some reasonable physical demands, being careful not to prescribe more components than the number of independent equations allow. Alternatively, one may resort to other physical theories, such as the theory of Dirac particles in order to construct a spin tensor out of Dirac spinors. In this work we will utilize the \weys fluid description of matter along with a supplementary structure. As mentioned in the introduction, the \weys fluid has been used in a number of interesting studies in \ect \cite{ref:weysstudystart} - \cite{ref:weysstudyend}. This model represents a fluid whose elements possess net (intrinsic) spin as well as stress-energy content. 

It is useful to construct the spin tensor via a second-rank tensor, $\tau_{\alpha\beta}$, as:
\begin{equation}
 \tau_{\alpha\beta}^{\;\;\;\;\gamma}= \tau_{\alpha\beta}u^{\gamma}\, \label{eq:secondspin}
\end{equation}
with $u^{\gamma}$ the 4-velocity of the fluid. The tensor $\tau_{\alpha\beta}$ is often known as the spin density. It is antisymmetric and is often subject to the restriction
\begin{equation}
 \tau_{\alpha\beta}u^{\beta}=0\,. \label{eq:frenkel}
\end{equation}
This last equation is often referred to as the Frenkel condition \cite{ref:frenkel}. It encodes a statement about the spacelike nature of spin. There is some debate on the suitability of enforcing the Frenkel condition when it comes to cosmological applications \cite{ref:frenkeldebate}. Strictly speaking it shall not be relevant for the calculations here as the study here is not within that realm.

In the simplest of scenarios the matter will be unpolarized. That is, the spins would be oriented randomly. This implies that the average of the spin density tensor would vanish; i.e. $\langle \tau_{\alpha\beta}\rangle=0$, as well as its gradients. However, there are spin contributions in (\ref{eq:goodeom}) which are quadratic in the spin, and it is generally not true that $\langle \tau_{\alpha\beta}\tau^{\alpha\beta}\rangle=0$. Hence, the quadratic contributions from spin in a macroscopic average will still contribute to the Einstein-Cartan field equations \cite{ref:hehl}. This therefore allows one to write equation (\ref{eq:goodeom}) for the \weys fluid as \cite{ref:hehl}, \cite{ref:gasperiniprl}
\begin{equation}
 G^{\mu\nu}\left({\tiny{\christoffel{\alpha}{\beta}{\gamma}}}\right)-\kappa^{2}s^2\left(-2 u^{\mu}u^{\nu}-g^{\mu\nu}\right)=\kappa \mathcal{T}^{\mu\nu}\,, \label{eq:ssquaredeqn}
\end{equation}
where here the notation $s^{2}:=\tau_{\alpha\beta}\tau^{\alpha\beta}$ has been employed. The vector $u^{\mu}$ represents the local 4-velocity of the fluid. We take it to have the same functional form as that of the ship. 

A few comments are in order before proceeding. First, the spin contribution, $s^{2}$, is in principle prescribable. However, reasonable physics dictates that the larger the particle content, the larger the $s^{2}$ contribution. Therefore it may be desirable to make $s^{2}$ proportional to the fluid energy density. Second, the stress-energy tensor of a fluid is algebraically incompatible with the symmetries required for the warp drive metrics. Therefore, one cannot simply use a perfect fluid (or even an anisotropic fluid) stress-energy tensor on the right-hand side of (\ref{eq:ssquaredeqn}). There must be some spinless auxiliary structure to the matter in order to bring the algebraic class of the right-hand side to compatibility with the left-hand side. In the analysis below we leave $\mathcal{T}^{\mu\nu}$ free. In other words, it is whatever is required in order to create the warp drive. It will in general be algebraically decomposable as Segre characteristic
\begin{equation}
\left[1,(1,1,1)\right] + \mbox{aux}\,,
\end{equation}
where ``aux'' represents the residual algebraic structure (non fluid structure) of the left-hand side of the equation (\ref{eq:ssquaredeqn}). That is, the net stress-energy content is that of the spin fluid plus any auxiliary matter required for algebraic compatibility. 

\section{The warp drive in \ect}\label{sec:warpdrive}

\subsection{The traditional warp drive}\label{sec:alcubierre}
We will first analyze the original warp drive of Alcubierre but within \ect \cite{ref:alcub}. This is arguably the most studied of such metrics. Its line element takes the form
\begin{equation}
 {\rm d}s^{2}=-{\rm d}t^{2} +\left[{\rm d}z -v_{\rm{s}}(t) f(x,y,z-z_{\rm{s}}(t))\,{\rm d}t\right]^{2} + {\rm d}x^{2} + {\rm d}y^{2}\,. \label{eq:alcumetric}
\end{equation}
Here the quantity $v_{\rm{s}}(t)$ represents the coordinate velocity of the ship, $v_{\rm{s}}(t)={\rm d}z_{\rm{s}}(t)/{\rm d}t$, so that the ship is moving in the $z$ direction. Due to the complexity of many of the expressions required for calculation we will consider the ship velocity to be constant. The fact that energy conditions may be respected even in an accelerating scenario can be deduced from the lemma below. The Christoffel-Einstein tensor, $G_{\mu\nu}\left({\tiny{\christoffel{\alpha}{\beta}{\gamma}}}\right)$, constructed out of this metric is presented in the Appendix.  The function $f(x,y,z-z_{\rm{s}}(t))$ is required to be ``top-hat'' like with a value of $0$ outside the warp bubble, and a value of $1$ inside. Specifically, Alcubierre chose
\begin{equation}
 f\left(r_{\rm{s}}(t)\right)=\frac{\tanh\left[\sigma\left(r_{\rm{s}}(t)+P\right)\right] - \tanh\left[\sigma\left(r_{\rm{s}}(t)-P\right)\right]}{2\tanh(\sigma P)}\,, \label{eq:alcuf}
\end{equation}
where $\Scale[0.90]{r_{\rm{s}}(t)=\left\{x^{2}+y^{2}+ \left[z-z_{\rm{s}}(t)\right]^{2}\right\}^{1/2}}$, $P$ is the ``radius'' of the warp bubble, and $\sigma$ is a parameter which controls how close $f\left(r_{\rm{s}}(t)\right)$ is to a true top-hat function. In this version of the warp drive there are contracting and expanding volume elements near the ship. However, it should be stressed that this is simply a by product of the metric (\ref{eq:alcumetric}). The contraction has little to do with the arbitrarily high velocity of the warp bubble. The ship does not reside in that region of the spacetime and, in fact, by a modification one may construct a similar warp drive without the contraction of the volume elements \cite{ref:noncontract}.

Staying within the paradigm of general relativity for the moment, it is easy to see that the spacetime generated by the metric (\ref{eq:alcumetric}) violates the weak energy condition (WEC). To see this let us consider observers in free-fall whose 4-velocity is given by
\begin{equation}
 [u^{\mu}]=\left[1,\,0,\,0, v_{s} f(r_{\rm{s}}(t))\right]\,. \label{eq:4vel}
\end{equation}
In accordance with the literature we will refer to such observers as Eulerian. As long as the observer is a test particle (meaning his/her stress-energy and, in \ecg also spin structure, may be neglected), this observer will correspond to one in free-fall.
One may calculate the following quantity relevant to the WEC:
\begin{equation} 
 G_{\mu\nu}u^{\mu}u^{\nu}=\kappa \mathcal{T}_{\mu\nu}u^{\mu}u^{\nu}=-\frac{v_{s}^{2}}{4} \left[ \left(\partial_{x}f\right)^{2} + \left(\partial_{y}f\right)^{2} \right]\,, \label{eq:alcuecond}
\end{equation}
using (\ref{eq:4vel}) and of course the Christoffel connection for $G_{\mu\nu}$, as we are currently working within general relativity. Note that the right-hand side of (\ref{eq:alcuecond}) is non-positive, and hence a negative energy density will be measured by the free-fall observers. The distribution of this energy density, $\mathcal{T}_{\mu\nu}u^{\mu}u^{\nu}$, is depicted in figure \ref{fig:alcurho}. One may minimize the volume of exotic matter required by selecting parameters in $f$ so that the function is as close to a top-hat as reasonably possible. However, then the derivatives in (\ref{eq:alcuecond}) become large, so although the volume of the WEC violating region is minimized, the severity of the violation in the region is increased.

\begin{figure}[h!t]
\begin{center}
\includegraphics[width=3in]{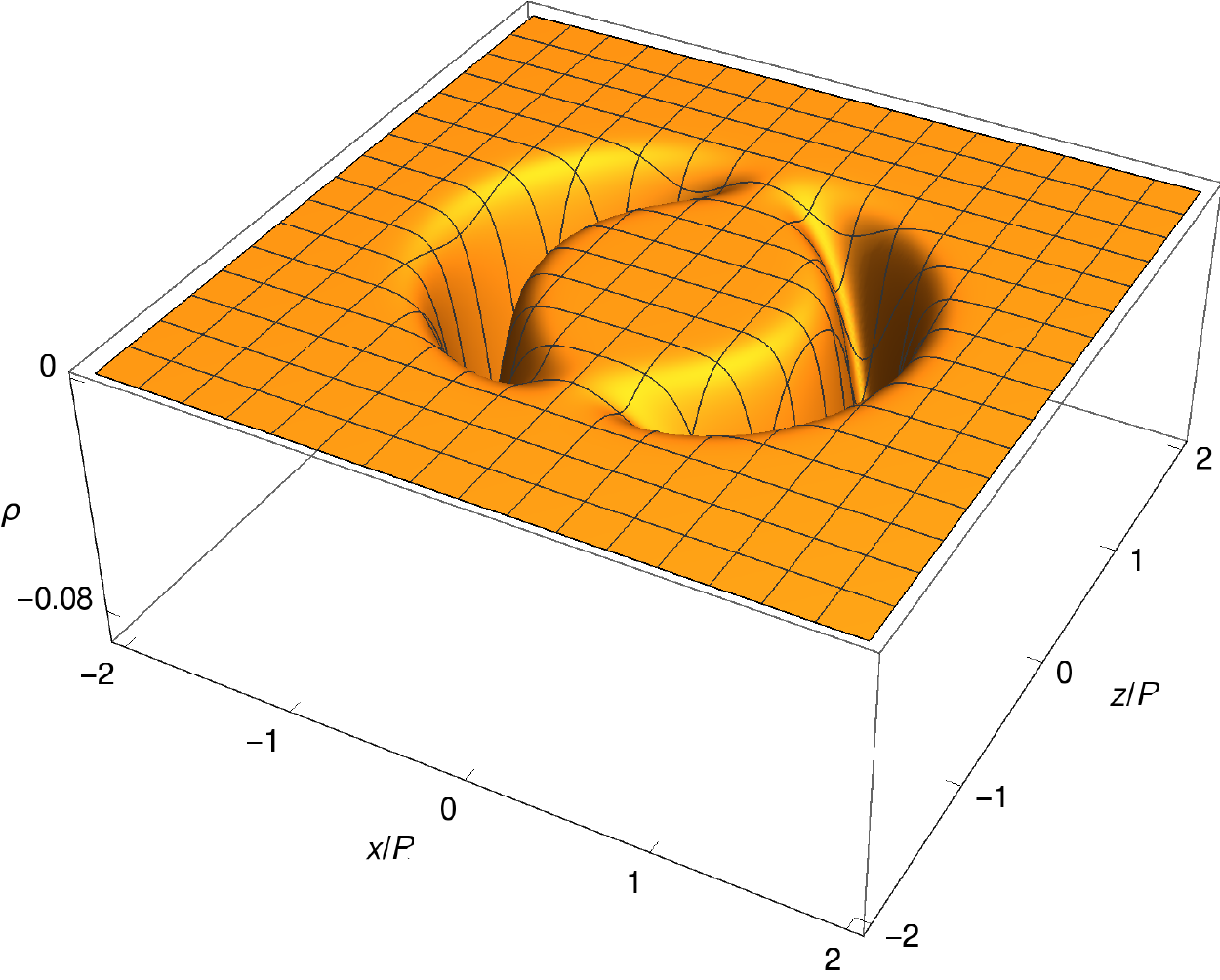}
\caption{{\small{A plot of the energy density as measured by freely falling observers in the original warp drive. This is the general relativity result on a $t=$constant slice and with the $y$ direction omitted. Note that there is energy condition violation wherever the plot is negative.}}}
\label{fig:alcurho}
\end{center}
\end{figure}

Next we wish to analyze the original warp drive within the context of \ectend. One potential issue is that, in general, the motion of free-falling particles in \ecg does not coincide with the geodesic equation. The torsion in general will couple to any spin the test particle may possess, altering its trajectory \cite{ref:pereirabook}. (The trajectory will be neither geodesic nor an autoparallel of the $U_{4}$ spacetime generally.) Strictly speaking, in order for the test particle (ship) to move according to the geodesic equation, the spin of the test particle must be small (required so that locally created torsion can be ignored, as required for a test particle). Also, ideally the contorsion of the spacetime should vanish in the vicinity of the test particle \cite{ref:handwrittenequations}. This issue will be circumvented by demanding that the solution possesses no matter (or at least very little matter) in the vicinity of the ship, and therefore the test particle is located in a region where there is no torsion, and hence general relativity, and the geodesic equation of test particles, holds. The matter field which is responsible for the warp bubble will, of course, not be in vacuum, but since it is not a test particle but part of the solution of the field equations, its distribution will be whatever is required by the field equations in order to produce metric (\ref{eq:alcumetric}) and eliminate energy condition violation. Its 4-velocity, although mimicking that of the ship in functional form, is not one of free-fall.

In \ecg there is the extra degree of freedom introduced from the presence of torsion, via the spin. This extra degree of freedom is manifest in $s^{2}$ and may be utilized in order to attempt to eliminate energy condition violation throughout the spacetime. We will concentrate our analysis on the WEC/null energy condition (NEC) specifically, which we will refer to as the WEC for simplicity, as by simple extension to null vectors the WEC can include the NEC. The WEC/NEC stipulates that for any timelike/null vector $v^{{\mu}}$, the following weak inequality must hold:
\begin{equation}
 \mathcal{T}_{{\mu}{\nu}}v^{{\mu}}v^{{\nu}} \geq 0\qquad \forall\quad v^{{\mu}}v_{{\mu}}=-1,\, 0 \,, \label{eq:wec}
\end{equation}
and this inequality will be rewritten using (\ref{eq:ssquaredeqn}) as
\begin{equation}
\mathcal{T}_{\mu\nu}v^{\mu}v^{\nu}=\left(\frac{1}{\kappa}G_{{\mu}{\nu}}\left({\tiny{\christoffel{\alpha}{\beta}{\gamma}}}\right) -\kappa s^2\left(-2 u_{{\mu}}u_{{\nu}}-g_{{\mu}{\nu}}\right) \right) {v}^{{\mu}}{v}^{{\nu}}\,. \label{eq:econd2a}
\end{equation}
We know that where the WEC is violated in general relativity  that the first term on the right-hand side will be negative. Therefore, we wish to prove that the remaining terms on the right-hand side can always be made sufficiently positive in order to negate the WEC violation in general relativity. The proof that the WEC can always be respected is facilitated by the following simple lemma:
\begin{lemma}
Define $\chi_{\mu\nu}:=\left(2 u_{{\mu}}u_{{\nu}}+g_{{\mu}{\nu}}\right)$ with $u^{\mu}$ a normalized future-pointing timelike vector field. Then, for $v^{\mu}$ a normalized timelike vector isochronous with $u^{\mu}$, the quantity $\chi_{\mu\nu}v^{\mu}v^{\nu}$ is strictly positive.
\end{lemma}
\begin{proof}
We have the expression
\begin{equation}
\chi_{\mu\nu}v^{\mu}v^{\nu}=2 u_{{\mu}}u_{{\nu}}v^{\mu}v^{\nu}+g_{{\mu}{\nu}}v^{\mu}v^{\nu}\,, \nonumber
\end{equation}
which due to the normalized nature of the vectors may be written as
\begin{equation}
2 \left(u_{{\mu}}v^{{\mu}}\right)^{2}-1\,. \label{eq:chimunu}
\end{equation}
The Cauchy-Schwarz inequality for timelike vectors states that
\begin{equation}
-g_{\mu\nu}u^{\mu}v^{\nu} \geq \sqrt{g_{\alpha\beta}u^{\alpha}u^{\beta}\: g_{\sigma\rho}v^{\sigma}v^{\rho}}\,, \nonumber
\end{equation}
which, due to the normalization of the vectors may be re-written as
\begin{equation}
-g_{\mu\nu}u^{\mu}v^{\nu} \geq 1\,. \nonumber
\end{equation}
The left-hand side of the weak inequality above is obviously positive, due to the timelike nature of $u^{\mu}$ and $v^{\mu}$ and the fact that they are isochronous (and of course also the fact it obeys the inequality). Therefore we conclude that $(u_{\mu}v^{\mu})^{2} \geq 1$. This then renders (\ref{eq:chimunu}) strictly positive and thus the assertion that the WEC (\ref{eq:econd2a}) may be made positive for sufficiently large $s^{2}$ is proven.
\end{proof}
\noindent The extension to show that $\chi_{\mu\nu}v^{\mu}v^{\nu}$ is positive for null $v^{\mu}$  is straightforward.

Although we have shown that the WEC can be satisfied, from a physical perspective it is desirable to accomplish this with the minimum amount of spin. That is, with the smallest $s^{2}$ allowable. For this we consider as before the left-hand side (l.h.s) of equation (\ref{eq:ssquaredeqn}) as:
\begin{equation}
 \mathcal{T}_{\hat{\mu}\hat{\nu}}=\frac{1}{\kappa}\left(\mbox{l.h.s.}\right)_{\alpha\beta} e_{\hat{\mu}}^{\;\,\alpha} e_{\hat{\nu}}^{\;\,\beta}\,, \label{eq:stressprojection}
\end{equation}
where we will perform the calculation in the orthonormal frame (indicated by hatted indices). Here $e_{\hat{\mu}}^{\;\alpha}$ indicate the components of the locally orthonormal tetrad, which we pick adapted to the motion as
\begin{equation}
\left[e^{\hat \mu}{}_{\alpha}\right] = \left[ \begin{array}{cccc}
 1 & 0 & 0 & 0 \\
 0 & 1 & 0 & 0 \\
 0 & 0 & 1 & 0 \\
 - v_s f & 0 & 0 & 1
\end{array} \right]\;,
\qquad
\left[e_{\hat \mu}{}^{\alpha}\right] = \left[ \begin{array}{cccc}
 1 & 0 & 0 & v_s f \\
 0 & 1 & 0 & 0 \\
 0 & 0 & 1 & 0 \\
 0 & 0 & 0 & 1
\end{array} \right]\,,
\end{equation}
although strictly speaking, since the analysis below will be in terms of invariants, any frame would be sufficient.

In order to make the analysis of the WEC more tractable, the following parameterization for the observer 4-velocity can be chosen, which respects the condition $v^{\hat{\mu}}v_{\hat{\mu}}=-1$ without loss of generality
\begin{equation} \label{eq:paramvel}
[\,v^{\hat\mu}\,] = [\, \cosh\beta,\;\sinh\beta\sin\theta\cos\phi,\;
\sinh\beta\sin\theta\sin\phi,\;\sinh\beta\cos\theta\,]\,,
\end{equation}
where
$\beta \in \mathbb{R}$, $\theta\in(0,\,\pi)$, and $\phi \in (-\pi,\,\pi)$.

To limit the amount of spin required in order to respect the WEC, first a specific trajectory is chosen in which to calculate the following quantity throughout the spacetime:
\begin{equation}
  \left(\frac{1}{\kappa}G_{\hat{\mu}\hat{\nu}}\left({\tiny{\christoffel{\alpha}{\beta}{\gamma}}}\right) -\kappa s^2\left(-2 u_{\hat{\mu}}u_{\hat{\nu}}-g_{\hat{\mu}\hat{\nu}}\right) \right) \tilde{v}^{\hat{\mu}}\tilde{v}^{\hat{\nu}}\,. \label{eq:econd2}
\end{equation}
Here the tilde indicates that $\tilde{v}^{\hat{\alpha}}$ is a specific 4-velocity. Equation (\ref{eq:econd2}) is, via (\ref{eq:ssquaredeqn}), the quantity which is required when the energy condition is to be calculated. This expression must be non-negative for all timelike ${v}^{\hat{\alpha}}$ in order for the WEC to be respected. The first term in (\ref{eq:econd2}) is completely determined from the metric (\ref{eq:alcumetric}), save for the 4-velocities, and since it is what goes into the WEC for general relativity, this term will in general be negative for warp drive metrics. If it is negative, then $s^{2}$ is to be set such that (\ref{eq:econd2}) vanishes everywhere. This needs to be done specifically for the vector $\tilde{v}^{\hat{\mu}}$ which produces the most severe negative result in general relativity at any given point in the spacetime. In other words, one chooses an $s^{2}$ so that energy condition inequality violation in general relativity, for the observer who measures this violation most severely at a certain point, is canceled by the spin terms which arise in \ecgend. This is to be done at every point in the spacetime, and may generally involve a different $\tilde{v}^{\hat{\mu}}$ vector at each spacetime point. It should be noted that in regions of the spacetime where the WEC would be violated within general relativity, this procedure will yield zero for the energy density as measured by the observer with 4-velocity $\tilde{v}^{\hat{\mu}}$ in \ecgend. Mathematically this is allowed. However, as mentioned previously, it is unphysical that spin should be present in the absence of matter. Hence, the procedure just described yields the absolute lower limit on $s^{2}$ which is capable of generating an energy condition respecting warp drive in \ectend. From a more physical perspective, one must actually increase the value of $s^{2}$, at least slightly, from this minimum value in order for the observer to measure a non-zero (and positive) energy density.

In principle, finding the minimum $s^{2}$ required (generally corresponding to the greatest WEC violation within general relativity) can be done via the standard extremization technique. One first sets the function (\ref{eq:econd2})$=0$, and solves for $s^{2}$ as a function of $\beta$, $\theta$ and $\phi$ (and, of course, the coordinates). Let us call this function $\tilde{S}:=s^{2}$; the value of $s^{2}$ which yields zero for (\ref{eq:econd2}). One then evaluates the gradients
\begin{equation}
 \frac{\partial \tilde{S}}{\partial \beta} =0,\quad \frac{\partial \tilde{S}}{\partial \theta} =0, \quad \frac{\partial \tilde{S}}{\partial \phi} =0\,, \label{eq:zerograd}
\end{equation}
and simultaneously solves these equations for $\beta$, $\theta$, and $\phi$. Now one will have critical values of $\beta(t,\,x,\,y,\,z)$, $\theta(t,\,x,\,y,\,z)$, and $\phi(t,\,x,\,y,\,z)$. At these values the function $\tilde{S}$ will be some sort of extremum. One wishes to find which ones are \emph{maxima} (as this corresponds to the vector $v^{\hat{\mu}}$ which produces the most negative result in the GR WEC). This is done by forming the Hessian
\begin{equation}
 H=\begin{bmatrix}
    \frac{\partial^{2} \tilde{S}}{\partial \beta^{2}}      & \frac{\partial^{2} \tilde{S}}{\partial \beta\,\partial\theta}&   \frac{\partial^{2} \tilde{S}}{\partial \beta \,\partial\phi}\\[0.1cm]
    \frac{\partial^{2} \tilde{S}}{\partial \theta \, \partial \beta}      & \frac{\partial^{2} \tilde{S}}{\partial\theta^{2}}&   \frac{\partial^{2} \tilde{S}}{\partial \theta \,\partial\phi}\\[0.1cm] 
    \frac{\partial^{2} \tilde{S}}{\partial \phi \, \partial \beta}      & \frac{\partial^{2} \tilde{S}}{\partial\phi\, \partial\theta}&   \frac{\partial^{2} \tilde{S}}{\partial \phi^{2}}
\end{bmatrix}_{|\mbox{\small{cp}}}\;, \label{eq:hessian}
\end{equation}
where ``cp'' indicates at the critical points, and by studying the determinants of the principal minors:
\begin{equation}
h_{1}:= \frac{\partial^{2} \tilde{S}}{\partial \beta^{2}}_{|\mbox{\small{cp}}}, \quad  h_{2}:=\begin{vmatrix}
    \frac{\partial^{2} \tilde{S}}{\partial \beta^{2}}      & \frac{\partial^{2} \tilde{S}}{\partial \beta \, \partial\theta}\\[0.1cm]
    \frac{\partial^{2} \tilde{S}}{\partial \theta \, \partial \beta}      & \frac{\partial^{2} \tilde{S}}{\partial\theta^{2}}
\end{vmatrix}_{|\mbox{\small{cp}}},  \quad h_{3}:=\left|H\right|_{|\mbox{\small{cp}}}\,.
\end{equation}
If $h_{1}$, $h_{2}$ and $h_{3}$ alternate between positive and negative at a critical point, then that critical point is a maximum. If a boundary is present it remains to be checked separately. 

The above procedure may be implemented in principle. However, in practice the warp drive spacetime is rather complicated, and the expression (\ref{eq:econd2}) is rather unwieldy. The lengths of the resulting expressions obscure any chance of reasonable analysis. Instead we are forced to resort to a somewhat simplified scenario. In other words we choose a special class of observers and calculate the spin required to respect the WEC inequality for this class of observers. Since the spacetime is $x$, $y$ symmetric, we will choose the class of observers boosted in the $z$ direction. In this scenario the 4 velocities (\ref{eq:paramvel}) simplify to
\begin{equation}
 [\,v^{\hat{\mu}}\,]
  = [\,\cosh \beta,\, 0,\, 0,\, \sinh \beta\,]\,. \label{eq:zvel}
\end{equation}
We will find the value of $\beta$ which produces the most serious WEC violation in general relativity, and construct $s^{2}$ such that it just cancels out this violation (keeping in mind the comment earlier that $s^{2}$ should be at least slightly larger than this). Solving for $\kappa^{2}s^{2}$ under this assumption yields
\begin{align}
 \kappa^{2}s^{2} = & - \frac{
    G_{\hat\mu\hat\nu}\left({\tiny{\christoffel{\alpha}{\beta}{\gamma}}}\right)
    v^{\hat{\mu}}v^{\hat{\nu}}
 }{
    2(u_{\hat{\alpha}}v^{\hat{\alpha}})^{2} - 1} \nonumber \\
= & \mbox{} \frac{v_{\rm{s}}^{2}}{4} \left( 2 - \sech(2\beta) \right)
    \left((\partial_{x}f)^{2} + (\partial_{y}f)^{2}\right)
    - \frac{v_{\rm{s}}}{2} \tanh(2\beta) \left(\partial^{2}_{x}f + \partial_{y}^{2}f\right)\,. \label{eq:spinexpr}
\end{align}
Setting the derivative with respect to $\beta$ equal to zero to find the extrema yields
\begin{equation}
 \sinh(2\beta)_{|\mbox{\small{cp}}} =
\frac{2(\partial_{x}^{2}f +\partial_{y}^{2}f)}{v_{\rm{s}}
\left((\partial_{x}f)^{2} + (\partial_{y}f)^{2}\right)}\,
\end{equation}
which gives the extrema of $\kappa^{2}s^{2}$ as
\begin{equation}
 \kappa^{2}s^{2} = \frac{v_{\rm{s}}^{2}}{2} \left(
(\partial_{x} f)^{2}+(\partial_{y} f)^{2}
- \sqrt{
  \frac14 \big( (\partial_{x} f)^{2}+(\partial_{y} f)^{2} \big)^2
+ \frac1{v_s^2} \big( \partial^{2}_{x}f + \partial^{2}_{y}f \big)^2
} \right) \,.
\end{equation}
The boundary points, $\beta\rightarrow \pm \infty$,
need to be checked independently and yield
\begin{equation}
 \kappa^{2}s^{2}{}_{|\beta\rightarrow \pm \infty} =
\frac{v_{\rm{s}}^{2}}{2}\left((\partial_{x}f)^{2} +(\partial_{y}f)^{2}\right)^{2}
\mp \frac{v_{\rm{s}}}{2} \left(\partial^{2}_{x}f +\partial^{2}_{y}f\right)\,. \label{eq:spinboundary}
\end{equation}
It turns out that the maxima occur on the boundary, $\beta\rightarrow \pm \infty$. Using the values of (\ref{eq:spinboundary}) gives the \emph{minimum} spin required in order to cancel out GR WEC violation for all observers boosted in the $z$ direction. Again we stress that in a physical situation, the spin should be at least slightly larger than this. We plot this spin in the vicinity of the ship in figure \ref{fig:alcussquared} for $v_{\rm{s}}=1.2$ and $\sigma=5$ inverse length units.
The resulting energy density, as measured by Eulerian observers using this value of $\kappa^{2}s^{2}$, is plotted in figure \ref{fig:energydens}. Note that this is everywhere non-negative. (It will be so for all observers by construction.)
\begin{figure}[h!t]
\begin{center}
\includegraphics[width=3in]{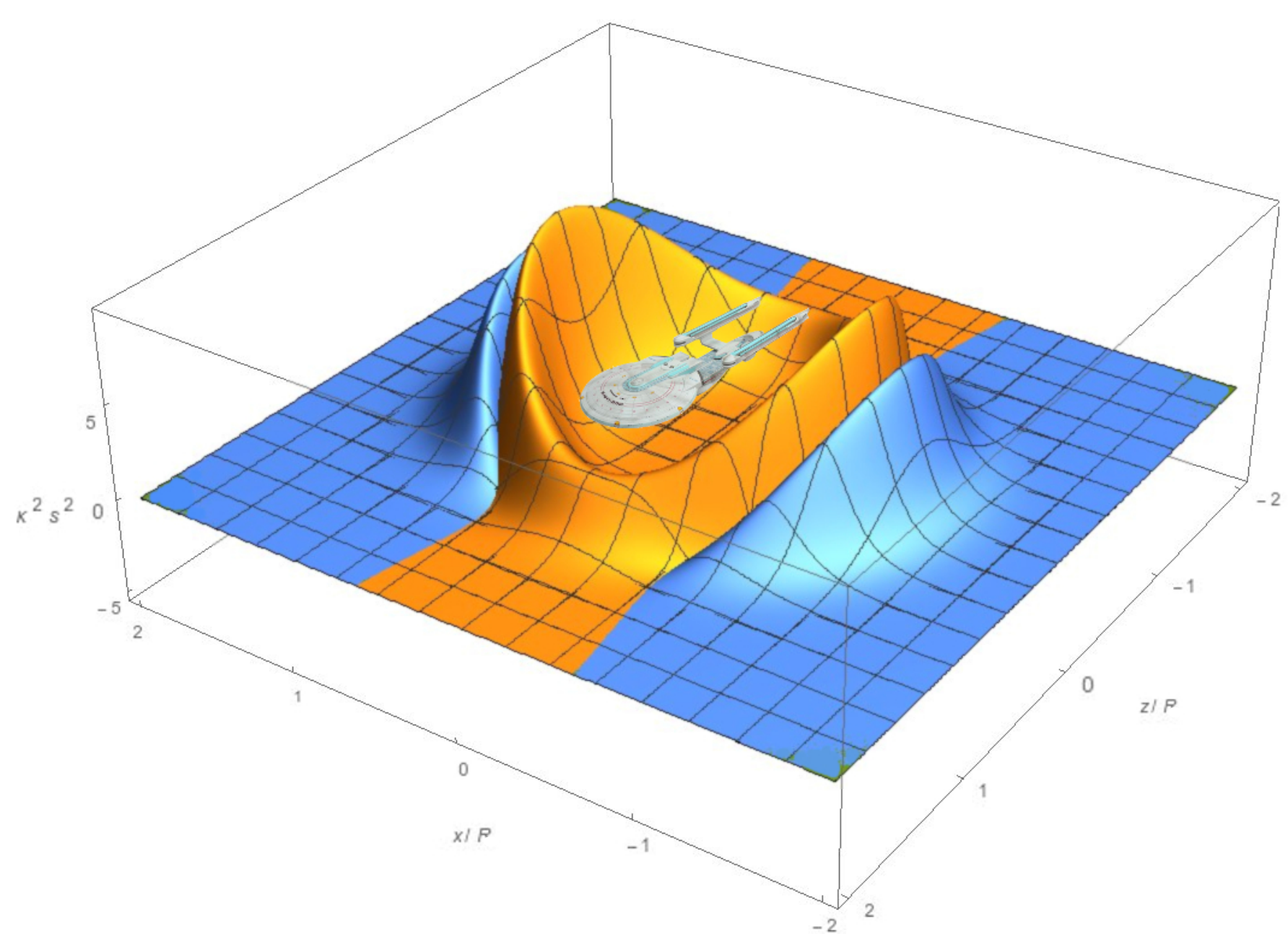}
\caption{{\small{A plot of the minimum spin required in order to have a WEC respecting warp drive for a ship velocity of $v_\mathrm{s}=1.2$.
The orange surface represents the result for $\beta\rightarrow +\infty$ and the blue the result for $\beta\rightarrow -\infty$. (See main text for details).}}}
\label{fig:alcussquared}
\end{center}
\end{figure}

\begin{figure}[h!t]
\begin{center}
\includegraphics[width=3in]{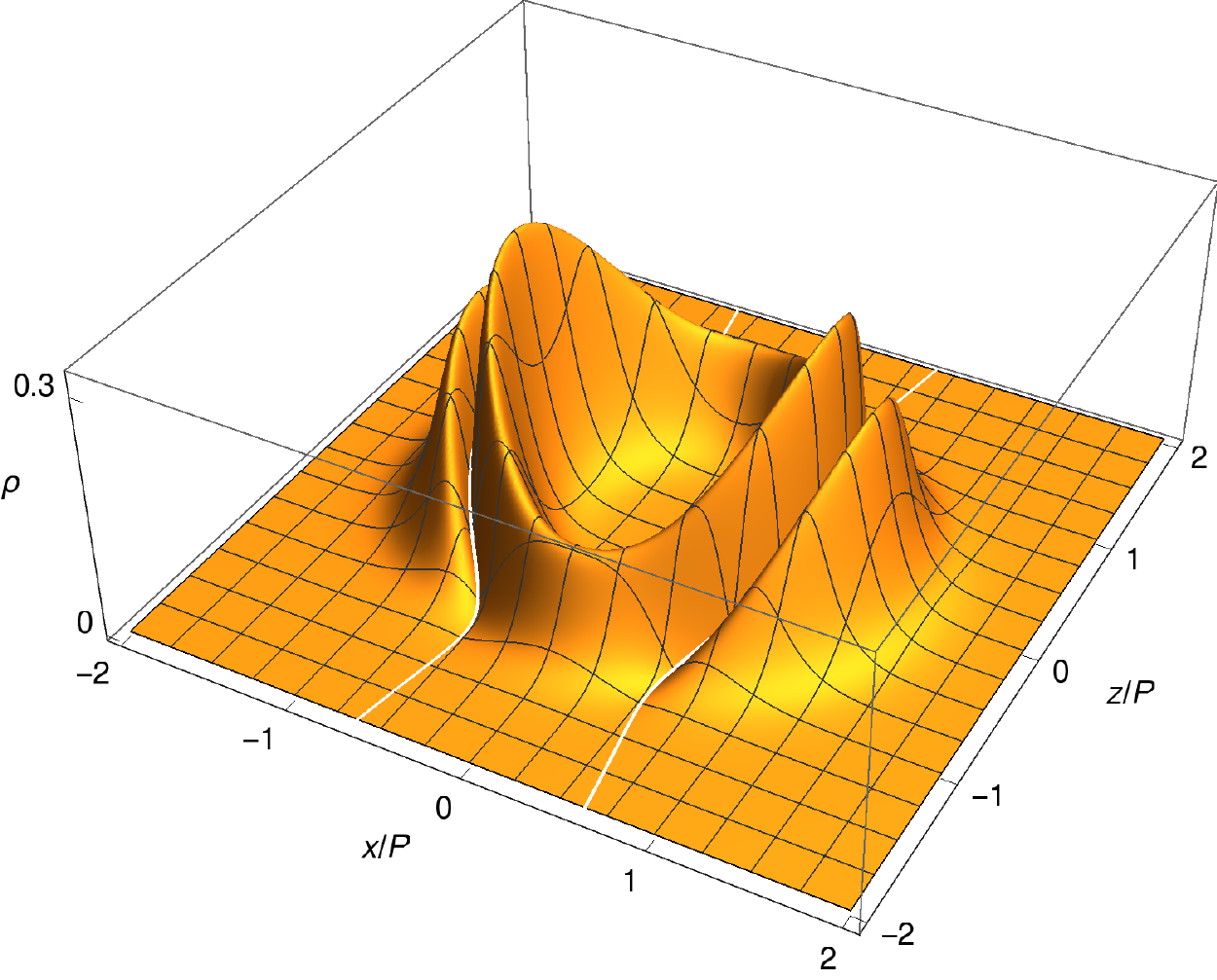}
\caption{{\small{The energy density as measured by Eulerian observers using the spin profile $s^{2}$ of figure \ref{fig:alcussquared}. (See main text for details).}}}
\label{fig:energydens}
\end{center}
\end{figure}

It might be interesting to consider the actual amount of spin required in order for this scheme to work. Let us consider the graph of figure \ref{fig:alcussquared} at the region with the highest spin density. Here $\kappa^{2}s^{2}$ is approximately $5$\,m$^{-2}$ at $x/P\approx 1$. Converting the corresponding value of $s$ to S.I. units yields an angular momentum density of approximately $\mathcal{L}\approx 3.6\times 10^{34}$kg/(m$\cdot$sec), or in units of $\hbar$, $\mathcal{L}\approx 3.5\times 10^{68}$ $\hbar/$m$^{3}$. For the sake of simplicity, and for a very crude approximation, let us assume that the spin density is sourced by particles whose spin is of the order $\hbar$. This spin density then corresponds to $\mathcal{N}=3.5\times 10^{68}$ of such particles per cubic meter. As another approximation, the energy density, as measured by Eulerian observers is given by $\mathcal{T}_{\mu\nu}u^{\mu}u^{\nu}$, where $u^{\mu}$ is of the form (\ref{eq:4vel}). This energy density is plotted in figure \ref{fig:energydens}. Near the maximum of this plot, also near $x/P =1$, this energy density is approximately $\rho_{0}=0.3$ inverse-square meter, which corresponds to $3.64\times 10^{43}$\, J/m$^{3}$, or a mass density equivalent of $4\times 10^{26}$\, kg/m$^{3}$. (As is often the case with exotic solutions in gravitational field theory, the devil is in the details.) If the field sourcing the spin is, for example, a monochromatic photon field\footnote{Coupling photons to torsion will, however, come at the expense of losing $U(1)$ gauge invariance.} (algebraic class of an anisotropic fluid) of frequency $\omega$, then this energy density corresponds to an electric field, $E$, of approximate magnitude
\begin{equation}
 E=\sqrt{\frac{2}{\epsilon_{0}} \hbar \omega \mathcal{N}} \approx 9\times 10^{22} \sqrt{\omega}\; \mbox{N/C}\,, \label{eq:elecfield}
\end{equation}
where the free-space permittivity is $\epsilon_{0}=8.85\times 10^{-12}\,$Farad/m. Of course, this is assuming that all of the stress-energy present at this point is due to the photon field itself, which for this illustrative purpose neglects the presence of the auxiliary matter discussed previously. $3.5\times 10^{68}$ photons per cubic meter, having an energy density of $3.64\times 10^{43}$ Joules per cubic meter, would possess a wavelength of approximately $1.8$ meters, which is not unreasonable, although somewhat large given the dimensions of the warp bubble for this example. Looking to other possible field sources, composite particles of high spin are unlikely candidates as they are found in unstable resonances. They also tend to carry too much energy due to their mass in order for them to be feasible in this scenario, as too much energy will tend increase the amount of spin required in order to equate the left and right-hand sides of the gravitational field equations. Fundamental high spin states are problematic in that it is difficult to describe point-like interactions of high spin in a consistent manner within a field theoretic framework, even in the massless case, although it seems within certain extensions of the standard theory it may be possible (see \cite{ref:klish} and \cite{ref:foto} for a summary of the issues and possible resolutions). One could perhaps utilize the arbitrarily high spin states allowed within the realm of string theory, although for low dimension these modes must be massive.

Even though the spin density is rather high, the overall amount of spin and energy required could be made much smaller. This was the motivation behind the modification to the warp drive by Van Den Broeck \cite{ref:vanden}, which modifies the geometry in such a way as to minimize the amount of matter required for the warp drive.

\subsection{The modified warp drive}\label{sec:vanden}
Here we will briefly discuss the Van Den Broeck warp drive \cite{ref:vanden} in light of \ectend. The method here mimics the analysis of the traditional warp drive above and so we only present the metric and the results. The Van Den Broeck warp drive minimizes the amount of exotic matter required by modifying the spacetime so that the volume in which the ship is located is bounded by a small area. In other words, a small warp bubble surrounds a throat leading to an approximately flat region with large volume. The line element is given by \cite{ref:vanden}
\begin{equation}
{\mathrm d}s^{2}= -{\mathrm d}t^{2}  +B^{2}(r_{\rm{s}})\left[\left({\mathrm d}z-v_{\rm{s}}(t) f(r_{\rm{s}}){\mathrm d}t\right)^{2} + {\mathrm d}x^{2} + {\mathrm{d}}y^{2} \right]\,. \label{eq:vandenmet}
\end{equation}
The Christoffel-Einstein tensor for this metric is also presented in the Appendix and again we consider the constant velocity scenario.
For the functions appearing in (\ref{eq:vandenmet}) we will make the same assumptions as in \cite{ref:vanden}. That is
\begin{align}
 B(r_{\rm{s}})= 1+\alpha, & \quad \text{for  } r_{\rm{s}} < \tilde{P}\,, \nonumber  \\ 
 1 < B(r_{\rm{s}}) \leq 1+ \alpha, & \quad \text{for  }  \tilde{P} \leq r_{\rm{s}} < \tilde{P} +\tilde{\Delta}\,,  \\
 B(r_{\rm{s}})=1, & \quad \text{for  } \tilde{P}+\tilde{\Delta} \leq r_{\rm{s}} \,, \nonumber 
\end{align}
with $P > \tilde{P}+\tilde{\Delta}$. The quantity $\tilde{\Delta}$ represents the coordinate thickness of the transition domain between the large volume inner region and the region which mimics the traditional warp drive. We use a function similar to (\ref{eq:alcuf}) ($B=1+f$), but with different parameters ($P\rightarrow \tilde{P}$), in order to model $B(r_{\rm{s}})$. Using the above metric one calculates the quantity (\ref{eq:econd2}) and, again considering longitudinally boosted observers (\ref{eq:zvel}), we set the value of $s$ by requiring that WEC violation is canceled for the most severe scenario. In this case it turns out that at certain points in the spacetime there are relevant extrema at the $\beta$ boundary, $\beta \rightarrow \pm\infty$, as well as at intermediate values of $\beta$ for other regions. We plot these values of $\kappa^{2}s^{2}$ in the vicinity of the warp bubble in figure \ref{fig:vdbspin}. The values chosen are as follows:
\begin{equation}
 P=3\,\mbox{fm}, \quad \tilde{P}=1\,\mbox{fm}, \quad \alpha =5, \quad \sigma =8\,\mbox{fm}^{-1}\,. \label{vdbparams}
\end{equation}
This actually corresponds to a tiny vessel, and the values are chosen only because they are useful for our purposes of analysis. This corresponds to a tiny warp bubble, whose inner-volume is admittedly not practical but it yields an idea of the spin densities required for a non-extreme case (one whose area-volume ratio is rather mild). We use femtometers here since the idea behind the modified warp drive is to have as small a WEC violating region as possible.

\begin{figure}[h!t]
\begin{center}
\includegraphics[width=3in, keepaspectratio=true]{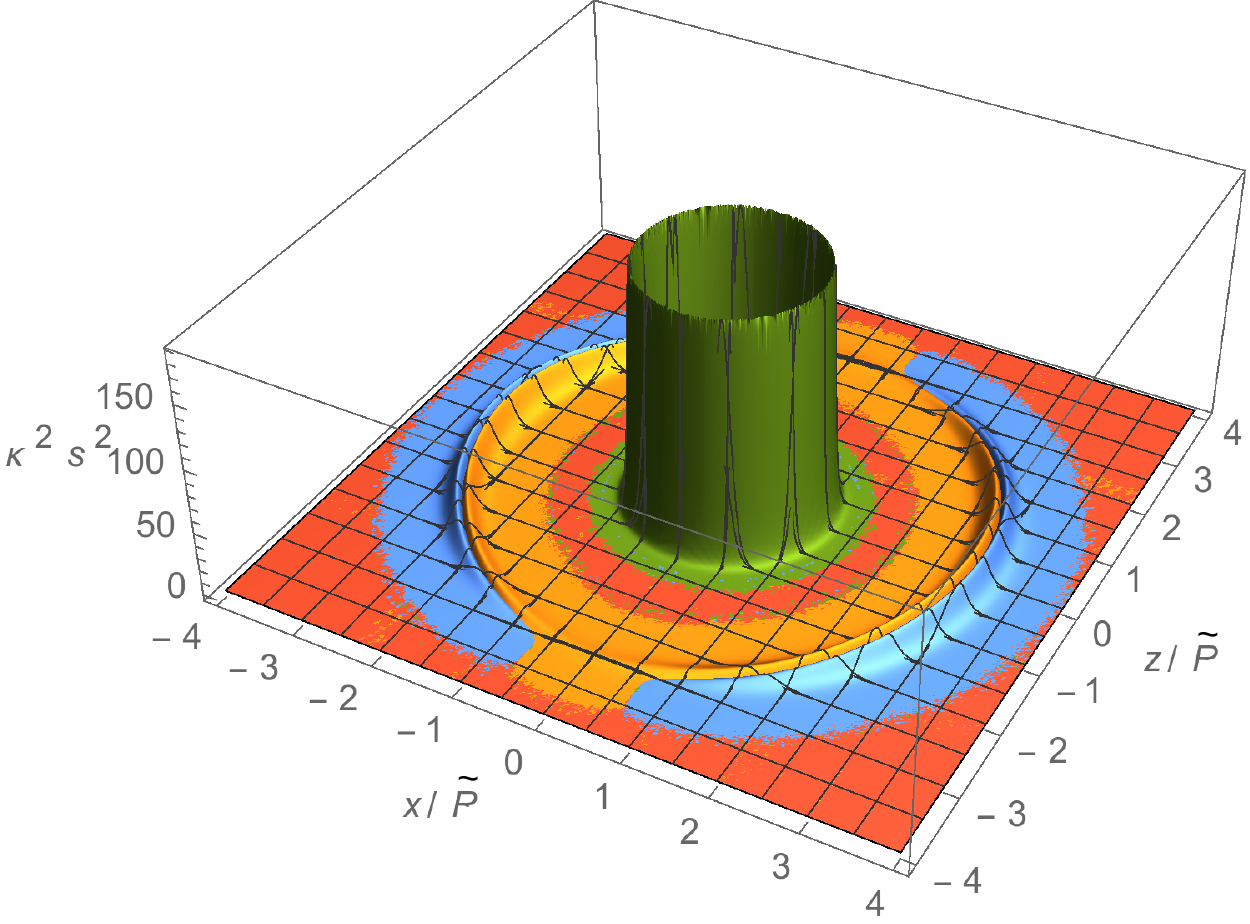}
\caption{{\small{A plot of the minimum spin required in order to have a WEC respecting modified warp drive for a ship velocity of $v_\mathrm{s}=1.2$.
The different colors represent the different values of $\beta$ which were used (boundary points and intermediate).}}}
\label{fig:vdbspin}
\end{center}
\end{figure}

From the plot in figure \ref{fig:vdbspin} it can be noted that the maximum value of $\kappa^{2}s^{2}$ is approximately 150. This in turn corresponds to an angular momentum density of approximately $\mathcal{L}\approx 1.9 \times 10^{35}$\, kg/(fm$\cdot$s) or $1.87 \times 10^{39}$\,$\hbar/$fm$^{3}$. This is of the order of $10^{84}$ spin-1 particles per cubic meter. Although the spin density is extremely large, one might minimize the net amount of spin required due to the design of this spacetime. However, the issue of exactly how to support this spin density is subject to similar comments as made for the traditional warp drive above.

Finally, we show the energy density as measured by Eulerian observers for this scenario in figure \ref{fig:vdbedense}. The approximate value here, $\rho_{\mbox{\tiny{max}}}\approx 0.6$\,fm$^{-2}$, corresponds to approximately $8\times 10^{11}$\, kg/fm$^{3}$ of mass density.

\begin{figure}[h!t]
\begin{center}
\includegraphics[width=3in, keepaspectratio=true]{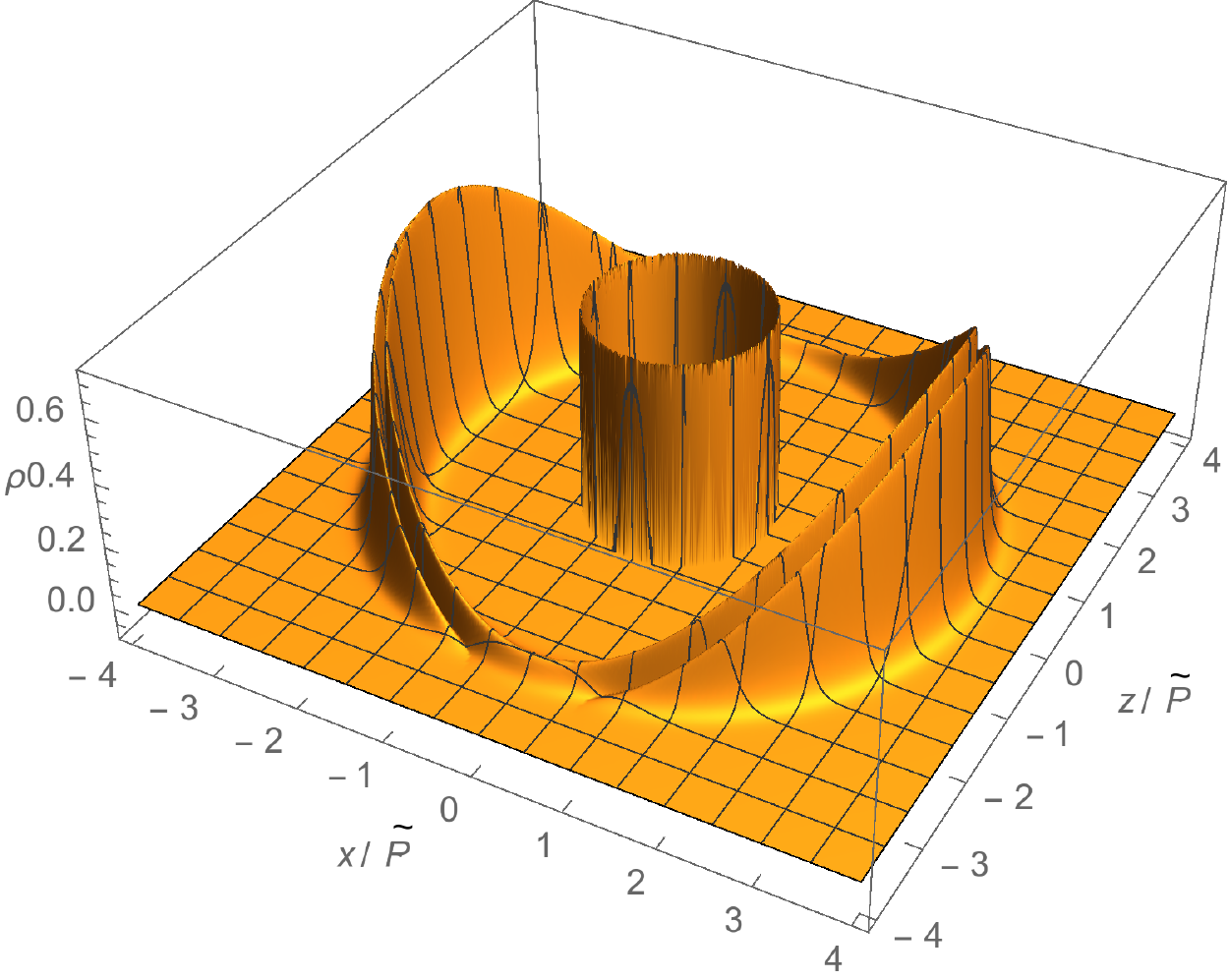}
\caption{{\small{The energy density as measured by Eulerian observers using the spin profile $s^{2}$ of figure \ref{fig:vdbspin} for the modified warp drive.}}}
\label{fig:vdbedense}
\end{center}
\end{figure}

In closing we should note that although in principle it is possible to create an energy condition respecting warp drive within \ectend, there are other issues with warp drive spacetimes which we did not address. For example, it has been shown that warp drive spacetimes contain an effective horizon, preventing the ship from possessing causal contact with points of the warp bubble when the effective speed becomes superluminal \cite{ref:krasnik}. This issue may be alleviated somewhat as it has been shown that parts of the bubble are still causally connected to the control region \cite{ref:causcon}. As well, when semiclassical effects are taken into account it has been shown that radiation build up in the warp drive spacetime may become large enough to render it unstable \cite{ref:semiclass}, at least at the semiclassical level. A full quantum analysis would require a theory of quantum gravity.

\section{{\small Concluding remarks}}\label{sec:conc}
It has been shown how, within the paradigm of \ectend, an energy condition respecting warp drive may exist, where no such counterpart is present in curvature-only general relativity. The Weyssenhoff fluid was utilized to calculate the spin density contribution, along with auxiliary structure to the matter which allows the stress energy tensor to be algebraically compatible with the warp drive spacetime. With the addition of spin as a source of gravity, the matter field supporting the warp bubble can indeed respect the weak energy condition inequality, and by extension the null energy condition. An attempt was made in order to minimize the amount of spin required in order to have WEC non-violation. For reasonable values of the parameters, we find that rather large spin angular momentum densities are required per cubic meter. The Van Den Broeck warp drive requires a much higher density, although in this case the spin distribution is located within a much smaller area, and therefore the net overall spin required may be less. Admittedly, these are rather large values and it is difficult to imagine how such spin densities may be achieved. However, the study does illustrate that in principle WEC violation may be alleviated in warp drive spacetimes within \ectend. It also serves to show that solutions which are considered exotic in general relativity may be less peculiar within theories where spacetime torsion exists as an extra degree of freedom.  The effects of torsion may therefore be non trivial in extreme scenarios.

\section*{{\small Acknowledgments}}
We are grateful to M. Sossich for stimulating discussions. AD would like to acknowledge the kind hospitality of FER, University of Zagreb, where part of this work was carried out. This work was partially supported by the VIF program of the University of Zagreb.

\PRLsep
\appendix
\setcounter{equation}{0}
\renewcommand\theequation{A.\arabic{equation}}

\section*{Appendix - The Christoffel-Einstein tensor ${\bm{G^{\mu\nu}\left({\tiny{\christoffel{\alpha}{\beta}{\gamma}}}\right)}}$}
We present here all the unique components of the Christoffel-Einstein tensor $G^{\mu\nu}\left({\tiny{\christoffel{\alpha}{\beta}{\gamma}}}\right)$ which appear in (\ref{eq:goodeom}), (\ref{eq:ssquaredeqn}), and which is used to calculate the weak/null energy condition violation in the general relativity limit, $G_{\mu\nu}\left({\tiny{\christoffel{\alpha}{\beta}{\gamma}}}\right)v^{\mu}v^{\nu}/\kappa$. We omit the explicit Christoffel connection dependence here.

\subsection*{The traditional warp drive}

The line element is
\begin{equation}
 {\rm d}s^{2}=-{\rm d}t^{2} +\left[{\rm d}z -v_{\rm{s}}(t) f(x,y,z-z_{\rm{s}}(t))\,{\rm d}t\right]^{2} + {\rm d}x^{2} + {\rm d}y^{2}\,.
\end{equation}

\begin{subequations} 
\romansubs
{\allowdisplaybreaks\begin{align}
G^{tt}=&\Scale[0.90]{-1/4\, \left( {v_{\rm{s}}}  
 \right) ^{2} \left(  \left( {\frac {\partial }{\partial x}}f \right) ^{2}+ \left( {\frac {\partial }{\partial y}}f
   \right) ^{2} \right) \,} \\[0.1cm]
G^{tz}=&\Scale[0.90]{-1/4\, \left( {v_{\rm{s}}}  
\right)  \left(  \left( {v_{\rm{s}}}   \right) ^{2}f   \left( {\frac {
\partial }{\partial x}}f   \right) ^{2}+ \left( 
{v_{\rm{s}}}   \right) ^{2}f
  \left( {\frac {\partial }{\partial y}}f
  \right) ^{2}+2\,{\frac {\partial ^{2}}{
\partial {y}^{2}}}f  +2\,{\frac {\partial ^{2}}{
\partial {x}^{2}}}f   \right)\,} \\[0.1cm]
G^{ty}=&\Scale[0.90]{1/2\, \left( {v_{\rm{s}}}  
\right) {\frac {\partial ^{2}}{\partial z\partial y}}f}  \, \\[0.1cm]
G^{tx}=&\Scale[0.90]{1/2\, \left( {v_{\rm{s}}}  
\right) {\frac {\partial ^{2}}{\partial z\partial x}}f  \,} \\[0.1cm]
G^{zz}=&\Scale[0.90]{-1/4\, \left( {v_{\rm{s}}}  
\right) ^{2} \left(  \left( {v_{\rm{s}}}   \right) ^{2} f ^{2}
\left( {\frac {\partial }{\partial x}}f  
\right) ^{2}+ \left( {v_{\rm{s}}}  \right) ^{2} \left( f \right)^{2}
\left( {\frac {\partial }{\partial y}}f  
\right) ^{2} \right.}  \nonumber\\
&\Scale[0.90]{+\left. 4\,f  {\frac {\partial ^{2}}{
\partial {y}^{2}}}f  +4\,f {\frac {\partial ^{2}}{\partial {x}^{2}}}f  +3\, \left( {\frac {\partial }{\partial x}}f   \right) ^{2}+3\, \left( {\frac {\partial }{\partial y}}f   \right) ^{2} \right) \, ,} \\[0.1cm]
G^{zy}=&\Scale[0.90]{1/2\, \left( {\frac {{\rm d}}{{\rm d}{t}}}{v_{\rm{s}}}   \right) {\frac {\partial }{\partial y}}f  +1/2\, \left( {v_{\rm{s}}}   \right) {\frac {\partial ^{2}}{\partial y\partial t}}f + \left( {\frac {\partial }{\partial z}}f   \right)  v_{\rm{s}} ^{2}{\frac {\partial }{\partial y}}f  +f   \left( v_{\rm{s}}\right) ^{2}{\frac { \partial ^{2}}{\partial z\partial y}}f  \,, } \\[0.1cm]
G^{zx}=&\Scale[0.90]{1/2\, \left( {\frac {{\rm d}}{{\rm d}{t}}}{v_{\rm{s}}}   \right) {\frac {\partial }{\partial x}}f  +1/2\, \left( {v_{\rm{s}}}   \right) {\frac {\partial ^{2}}{\partial x\partial t}}f
 + \left( {\frac {\partial }{\partial z}}f
  \right)  v_{\rm{s}} ^{2}{\frac {\partial }{\partial x}}f
 +f   \left( {v_{\rm{s}}}   \right) ^{2}{\frac {
\partial ^{2}}{\partial z\partial x}}f  \, ,} \\[0.1cm]
G^{yy}=&\Scale[0.90]{1/4\, \left( {\frac {\partial }{\partial y}}f  
\right) ^{2} \left( {v_{\rm{s}}}   \right) ^{2}-f   v_{\rm{s}} ^{2}{\frac {\partial ^{2
}}{\partial {z}^{2}}}f  -1/4\, \left( {\frac {
\partial }{\partial x}}f   \right) ^{2} \left( {
v_{\rm{s}}}   \right) ^{2}} \nonumber \\
&\Scale[0.90]{-\left( {\frac {{\rm d}}{{\rm d}{t}}}{v_{\rm{s}}}   \right) {\frac {\partial }{\partial z}}f - \left( {v_{\rm{s}}}  
\right) {\frac {\partial ^{2}}{\partial z\partial t}}f  - \left( {\frac {\partial }{\partial z}}f   \right) ^{2} \left( {v_{\rm{s}}}
  \right) ^{2}\,,} \\[0.1cm]
G^{yx}=&\Scale[0.90]{1/2\, \left( {\frac {\partial }{\partial y}}f  
\right)  \left( {v_{\rm{s}}}  
\right) ^{2}{\frac {\partial }{\partial x}}f \, ,} \\[0.1cm]
G^{xx}=&\Scale[0.90]{1/4\, \left( {\frac {\partial }{\partial x}}f  
\right) ^{2} \left( {v_{\rm{s}}}   \right) ^{2}-f   v_{\rm{s}} ^{2}{\frac {\partial ^{2
}}{\partial {z}^{2}}}f  -1/4\, \left( {\frac {
\partial }{\partial y}}f   \right) ^{2} \left( {
v_{\rm{s}}}   \right) ^{2} } \nonumber \\
&\Scale[0.90]{- \left( {\frac {{\rm d}}{{\rm d}{t}}}{v_{\rm{s}}}   \right) {\frac {\partial }{\partial z}}f  - \left( {v_{\rm{s}}}  
\right) {\frac {\partial ^{2}}{\partial z\partial t}}f  - \left( {\frac {\partial }{\partial z}}f   \right) ^{2} \left( {v_{\rm{s}}}
  \right) ^{2}\, .}
\end{align}}
\end{subequations}

\subsection*{The modified warp drive}

The line element is
\begin{equation}
 {\mathrm d}s^{2}= -{\mathrm d}t^{2}  +B^{2}(r_{\rm{s}})\left[\left({\mathrm d}z-v_{\rm{s}} f(r_{\rm{s}}){\mathrm d}t\right)^{2} + {\mathrm d}x^{2} + {\mathrm{d}}y^{2} \right]\,.
\end{equation}
Due to the complexity of the resulting expressions, we set $v_{\rm{s}}=\mbox{const.}$ right away here.

\begin{subequations} 
\romansubs
{\allowdisplaybreaks\begin{align}
G^{tt}=& \Scale[0.90]{-\frac{1}{4 B^{4}}\,\left[-12\, B ^{2}
 f ^{2} \left( {\frac {\partial }{\partial z}}B   \right) ^{2}{v_{\rm{s}}}^{2}-8\, \left( B
   \right) ^{3}f   \left( {\frac {\partial }{\partial z}}B   \right)  \left( {\frac {\partial }{\partial z}}f   \right) {v_{\rm{s}}}^{2}+ \left( {\frac {\partial }{
\partial y}}f   \right) ^{2} B ^{4}{v_{\rm{s}}}^{2} \right.} \nonumber \\
&\Scale[0.90]{\left. + B ^{4} \left( {\frac {\partial }{\partial x}}f   \right) ^{2}{v_{\rm{s}}}^{2}-24\,v_{\rm{s}}\,f   B ^{2} \left( {
\frac {\partial }{\partial z}}B   \right) {\frac {\partial }{\partial t}}B  -8\, \left( B
   \right) ^{3} \left( {\frac {\partial }{\partial t}}B   \right)  \left( {\frac {
\partial }{\partial z}}f   \right) v_{\rm{s}}-12\, B ^{2} \left( {\frac {
\partial }{\partial t}}B   \right) ^{2} \right.} \nonumber \\
&\Scale[0.90]{ \left. +8\,B  {\frac {\partial ^{2}}{\partial {y}^{2}}}B
  +8\,B  {\frac {\partial ^{2}}{\partial {z}^{2}}}B  -4\, \left( {\frac {
\partial }{\partial x}}B   \right) ^{2}+8\,B  {\frac {\partial ^{2}}{\partial {x}^{2}}}B
  -4\, \left( {\frac {\partial }{\partial y}}B   \right) ^{2}-4\, \left( {\frac {\partial }{
\partial z}}B   \right) ^{2}\right]\, ,} \\[0.1cm]
G^{tz}=&\Scale[0.90]{\frac{1}{4 
 B ^{4} }\,\left[12\, B ^{2}
 f ^{3} \left( {\frac {\partial }{\partial z}}B   \right) ^{2}{v_{\rm{s}}}^{3}+8\, \left( B 
   \right) ^{3} f ^{2} \left( {\frac {\partial }{
\partial z}}B   \right)  \left( {\frac {\partial }{\partial z}}f   \right) {v_{\rm{s}}}^{
3}-f   \left( {\frac {\partial }{\partial y}}f   \right) ^{2} \left( B  
 \right) ^{4}{v_{\rm{s}}}^{3} \right.} \nonumber \\
 &\Scale[0.90]{\left. - B ^{4}f   \left( {\frac {\partial }{
\partial x}}f   \right) ^{2}{v_{\rm{s}}}^{3}+24\, B ^{2} \left( f 
  \right) ^{2} \left( {\frac {\partial }{\partial z}}B   \right)  \left( {\frac {\partial }{\partial t}}B
   \right) {v_{\rm{s}}}^{2}+8\, B ^{3}f   \left( {\frac {
\partial }{\partial t}}B   \right)  \left( {\frac {\partial }{\partial z}}f   \right) 
{v_{\rm{s}}}^{2} \right.} \nonumber \\
&\Scale[0.90]{ \left. +12\, B ^{2}f   \left( {\frac {\partial }{\partial t}}B 
  \right) ^{2}v_{\rm{s}}-2\, B ^{2}v_{\rm{s}}\,{\frac {\partial ^{2}}{\partial {y}^{2}}}f
  -8\,B  v_{\rm{s}}\,f  {\frac {\partial ^{2}}{\partial {y}^{2}}}B
  \right.} \nonumber \\
  & \Scale[0.90]{\left. +4\, \left( {\frac {\partial }{\partial x}}B   \right) ^{2}v_{\rm{s}}\,f  -8\,B  v_{\rm{s}}\,f  {\frac {\partial ^{2}}{\partial {x}^{2}}}B  +4\, \left( {\frac {\partial }{\partial y}}B 
   \right) ^{2}v_{\rm{s}}\,f  -4\, \left( {\frac {\partial }{\partial z}}B   \right) ^{2}{
v_{\rm{s}}}\,f  -6\,B  v_{\rm{s}}\, \left( {\frac {\partial }{\partial x}}f  
 \right) {\frac {\partial }{\partial x}}B \right.} \nonumber \\
 &\Scale[0.90]{ \left. -6\,B  v_{\rm{s}}\, \left( {\frac {\partial }{
\partial y}}f   \right) {\frac {\partial }{\partial y}}B  -2\, B ^{2}v_{\rm{s}}\,{\frac {\partial ^{2}}{\partial {x}^{2}}}f  +8\,B  {\frac {\partial ^{2}}{\partial z\partial t}}B  -8\,
 \left( {\frac {\partial }{\partial t}}B   \right) {\frac {\partial }{\partial z}}B  \right]\, ,} \\[0.1cm]
 G^{ty}=&\Scale[0.90]{\frac{1}{2 B ^{4}}\,\left[4\,B  v_{\rm{s}}\,f  {\frac {\partial ^{2}}{\partial z\partial y}}B  + B ^{2}v_{\rm{s}}\,{\frac {\partial ^{2}}{\partial z\partial y}}f  -
4\, \left( {\frac {\partial }{\partial z}}B   \right) v_{\rm{s}}\,f  {\frac {\partial }{
\partial y}}B  +B  v_{\rm{s}}\, \left( {\frac {\partial }{\partial y}}f  
 \right) {\frac {\partial }{\partial z}}B \right.} \nonumber \\
 &\Scale[0.90]{ \left. +2\,B  v_{\rm{s}}\, \left( {\frac {\partial }{
\partial z}}f   \right) {\frac {\partial }{\partial y}}B  +4\,B  {
\frac {\partial ^{2}}{\partial y\partial t}}B  -4\, \left( {\frac {\partial }{\partial t}}B  
 \right) {\frac {\partial }{\partial y}}B  \right]\,,} \\[0.1cm]
 G^{tx}=& G^{ty} \quad x \leftrightarrow y \, , \\[0.1cm]
 G^{zz}=&\Scale[0.90]{-\frac{1}{4 B ^{6} }\,\left[-8\, f ^{2}
 B ^{5} \left( {\frac {\partial }{\partial t}}B   \right)  \left( {
\frac {\partial }{\partial z}}f   \right) {{v_{\rm{s}}}}^{3}-12\, f ^{2} B ^{4} \left( {\frac {\partial }{\partial t}}B   \right) ^{2}{{v_{\rm{s}}}}^{2}+4\,
 B ^{4}f  {{v_{\rm{s}}}}^{2}{\frac {\partial ^{2}}{\partial {y}^{2}}}f \right.} \nonumber \\
 &\Scale[0.90]{\left.  +8\, B ^{3} f ^{2}{{v_{\rm{s}}}}^{2}{
\frac {\partial ^{2}}{\partial {y}^{2}}}B  -4\, B ^{2}{{v_{\rm{s}}}}^{2} 
f ^{2} \left( {\frac {\partial }{\partial x}}B   \right) ^{2}+8\, 
B ^{3} f ^{2}{{v_{\rm{s}}}}^{2}{\frac {\partial ^{2}}{\partial {x}^{2}}}B
  -4\, f ^{2} B ^{2} \left( {\frac {
\partial }{\partial y}}B   \right) ^{2}{{v_{\rm{s}}}}^{2} \right.} \nonumber \\
& \Scale[0.90]{\left. +4\, B ^{4}f 
 {{v_{\rm{s}}}}^{2}{\frac {\partial ^{2}}{\partial {x}^{2}}}f  -12\, \left( f  
 \right) ^{4} B ^{4} \left( {\frac {\partial }{\partial z}}B   \right) ^{2}{{
v_{\rm{s}}}}^{4}+ f ^{2} \left( {\frac {\partial }{\partial y}}f   \right) ^{2}
 B ^{6}{{v_{\rm{s}}}}^{4}+ f ^{2} \left( {\frac {\partial }{
\partial x}}f   \right) ^{2} B ^{6}{{v_{\rm{s}}}}^{4} \right.} \nonumber \\
& \Scale[0.90]{\left. -8\, f ^{3} 
B ^{5} \left( {\frac {\partial }{\partial z}}B   \right)  \left( {\frac {\partial }{\partial z}}f 
  \right) {{v_{\rm{s}}}}^{4}-24\, f ^{3} B ^{4}
 \left( {\frac {\partial }{\partial t}}B   \right)  \left( {\frac {\partial }{\partial z}}B   \right) {{v_{\rm{s}}}}^{3}+12\, B ^{3}{{v_{\rm{s}}}}^{2}f   \left( {\frac {\partial }{\partial x}}B   \right) {
\frac {\partial }{\partial x}}f \right.} \nonumber \\
&\Scale[0.90]{ \left. +12\,f   \left( {\frac {\partial }{\partial y}}f   \right)  
B ^{3} \left( {\frac {\partial }{\partial y}}B  
 \right) {{v_{\rm{s}}}}^{2}+8\, B ^{3}{\frac {\partial ^{2}}{\partial {t}^{2}}}B  
+4\, \left( {\frac {\partial }{\partial x}}B   \right) ^{2}-4\,B  {\frac {\partial ^{2}}{
\partial {x}^{2}}}B  +4\, \left( {\frac {\partial }{\partial y}}B   \right) ^{2} \right.} \nonumber \\
&\Scale[0.90]{ \left. -4\,B
  {\frac {\partial ^{2}}{\partial {y}^{2}}}B  -4\, \left( {\frac {\partial }{\partial z}}B
   \right) ^{2}+24\,{v_{\rm{s}}}\,f   B ^{2} \left( {\frac 
{\partial }{\partial z}}B   \right) {\frac {\partial }{\partial t}}B  +16\, B ^{2} 
f ^{2} \left( {\frac {\partial }{\partial z}}B   \right) ^{2}{{v_{\rm{s}}}}^{2} \right.} \nonumber \\
& \Scale[0.90]{\left. +8\, \left( B 
  \right) ^{3} \left( {\frac {\partial }{\partial z}}B   \right)  \left( {\frac {\partial }{\partial t}}f
   \right) {v_{\rm{s}}}+4\, B ^{2} \left( {\frac {\partial }{\partial t}}B   \right) ^{2}+8\, B ^{3}f   \left( {\frac {\partial }{
\partial z}}B   \right)  \left( {\frac {\partial }{\partial z}}f   \right) {{v_{\rm{s}}}}^{
2}+3\, B ^{4} \left( {\frac {\partial }{\partial x}}f   \right) ^{2}{{v_{\rm{s}}}}^{2} \right.} \nonumber \\
&\Scale[0.90]{ \left. +3\, \left( {\frac {\partial }{\partial y}}f   \right) ^{2} B ^{4}{{v_{\rm{s}}}}^{2}\right]\, ,}\\[0.1cm]
G^{zy}=&\Scale[0.90]{\frac{1}{2 B ^{6} }\,\left[4\, B ^{3} f ^{2}{{v_{\rm{s}}}}^{2}{\frac {\partial 
^{2}}{\partial z\partial y}}B  +2\, B ^{4}f  
{{v_{\rm{s}}}}^{2}{\frac {\partial ^{2}}{\partial z\partial y}}f  -4\, f ^{2} B ^{2} \left( {\frac {\partial }{\partial y}}B   \right)  \left( {\frac {\partial }{\partial z}}B   \right) {{v_{\rm{s}}}}^{
2} \right.} \nonumber \\
& \Scale[0.90]{\left. +4\, B ^{3}{{v_{\rm{s}}}}^{2}f   \left( {\frac {\partial }{\partial z}}B   \right) {\frac {\partial }{\partial y}}f  +2\,f   B ^{3} \left( {\frac {\partial }{\partial y}}B
   \right)  \left( {\frac {\partial }{\partial z}}f   \right) {{v_{\rm{s}}}}^{2}+2\, \left( {\frac 
{\partial }{\partial y}}f   \right)  B ^{4} \left( {\frac {\partial }{
\partial z}}f   \right) {{v_{\rm{s}}}}^{2} \right.} \nonumber \\
&\Scale[0.90]{ \left. + B ^{4}{v_{\rm{s}}}\,{\frac {\partial ^{2}
}{\partial y\partial t}}f  +4\, B ^{3}{v_{\rm{s}}}\,f  {
\frac {\partial ^{2}}{\partial y\partial t}}B  -4\,f   B 
^{2} \left( {\frac {\partial }{\partial y}}B   \right)  \left( {\frac {\partial }{\partial t}}B   \right) {v_{\rm{s}}}+3\, 
B ^{3} \left( {\frac {\partial }{\partial t}}B   \right) {v_{\rm{s}}}\,{\frac {\partial }{\partial y}}f \right.} \nonumber \\
 & \Scale[0.90]{\left. -2\,B  {\frac {\partial ^{2}}{\partial z\partial y}}B  +4\, \left( {\frac {
\partial }{\partial z}}B   \right) {\frac {\partial }{\partial y}}B  \right]\,} \\[0.1cm]
G^{zx}=&G^{zy} \quad x \leftrightarrow y \, , \\[0.1cm]
G^{yy}=&\Scale[0.90]{-\frac{1}{4 B^{6}}\,\left[4\, B ^{4}f  {{v_{\rm{s}}}}^{2}{\frac {\partial ^{2}}{
\partial {z}^{2}}}f  +8\, B ^{3} f ^{2}{
{v_{\rm{s}}}}^{2}{\frac {\partial ^{2}}{\partial {z}^{2}}}B  +4\, B ^{2} \left( f
   \right) ^{2} \left( {\frac {\partial }{\partial z}}B   \right) ^{2}{{v_{\rm{s}}}}^{2}+20\,
 B ^{3}f   \left( {\frac {\partial }{\partial z}}B   \right)  \left( {\frac {\partial }{\partial z}}f   \right) {{v_{\rm{s}}}}^{2} \right.} \nonumber \\
 & \Scale[0.90]{\left. - \left( {\frac {\partial }{
\partial y}}f   \right) ^{2} B ^{4}{{v_{\rm{s}}}}^{2}+ B ^{4} \left( {\frac {\partial }{\partial x}}f   \right) ^{2}{{v_{\rm{s}}}}^{2}+4\, B ^{4} \left( {\frac {\partial }{\partial z}}f
   \right) ^{2}{{v_{\rm{s}}}}^{2}+4\, B ^{4}{v_{\rm{s}}}\,{\frac {\partial ^{2}}
{\partial z\partial t}}f  +16\, B ^{3}{v_{\rm{s}}}\,f  {
\frac {\partial ^{2}}{\partial z\partial t}}B \right.} \nonumber \\
& \Scale[0.90]{\left. +8\,{v_{\rm{s}}}\,f   B ^{2} \left( {\frac {\partial }{\partial z}}B   \right) {\frac {\partial }{\partial t}}B  +12\, B ^{3} \left( 
{\frac {\partial }{\partial t}}B   \right)  \left( {\frac {\partial }{\partial z}}f  
 \right) {v_{\rm{s}}}+8\, B ^{3} \left( {\frac {\partial }{\partial z}}B  
 \right)  \left( {\frac {\partial }{\partial t}}f   \right) {v_{\rm{s}}}  +8\, \left( B  
 \right) ^{3}{\frac {\partial ^{2}}{\partial {t}^{2}}}B\right.}  \nonumber\\
 &\Scale[0.88]{\left.  +4\, B ^{2} \left( {
\frac {\partial }{\partial t}}B   \right) ^{2}-4\,B  {\frac {\partial ^{2}}{\partial {z}^{2}}}B
  +4\, \left( {\frac {\partial }{\partial x}}B   \right) ^{2}-4\,B  {
\frac {\partial ^{2}}{\partial {x}^{2}}}B  -4\, \left( {\frac {\partial }{\partial y}}B  
 \right) ^{2}+4\, \left( {\frac {\partial }{\partial z}}B   \right) ^{2}\right]\, ,} \\[0.1cm]
 G^{yx}=&\Scale[0.90]{\frac{1}{2  B ^{6}}\,\left[B ^{4}{{v_{\rm{s}}}}
^{2} \left( {\frac {\partial }{\partial x}}f  
 \right) {\frac {\partial }{\partial y}}f  +4\,
 \left( {\frac {\partial }{\partial y}}B  
 \right) {\frac {\partial }{\partial x}}B  \right.} { \left. -2\,B
  {\frac {\partial ^{2}}{\partial y\partial x}}B
 \right]\, ,} \\[0.1cm]
 G^{xx}=&G^{yy} \quad x \leftrightarrow y \, .
\end{align}}
\end{subequations}
\PRLsep

\vspace{-0.080cm}

%%\vspace{0.5cm}
%%\newpage
\linespread{0.6}
\bibliographystyle{unsrt}

%%\PRLsep
\end{document}